\documentclass[preprint]{elsarticle}

\usepackage{lineno,hyperref}
\usepackage{graphicx,amsmath,amssymb,amsthm,bm}
\modulolinenumbers[5]

\journal{}

\newtheorem{definition}{Definition}
\newtheorem{proposition}{Proposition}
\newtheorem{lemma}{Lemma}
\newtheorem{theorem}{Theorem}

\begin{document}

\begin{frontmatter}

\title{Properties of zero-determinant strategies in multichannel games}

\author{Masahiko Ueda\corref{mycorrespondingauthor}}
\ead{m.ueda@yamaguchi-u.ac.jp}
\address{Graduate School of Sciences and Technology for Innovation, Yamaguchi University, Yamaguchi 753-8511, Japan}

\begin{abstract}
Controlling payoffs in repeated games is one of the important topics in control theory of multi-agent systems.
Recently proposed zero-determinant strategies enable players to unilaterally enforce linear relations between payoffs.
Furthermore, based on the mathematics of zero-determinant strategies, regional payoff control, in which payoffs are enforced into some feasible regions, has been discovered in social dilemma situations.
More recently, theory of payoff control was extended to multichannel games, where players parallelly interact with each other in multiple channels.
However, the existence of payoff-controlling strategies in multichannel games seems to require the existence of payoff-controlling strategies in some channels, and properties of zero-determinant strategies specific to multichannel games are still not clear.
In this paper, we elucidate properties of zero-determinant strategies in multichannel games.
First, we relate the existence condition of zero-determinant strategies in multichannel games to that of zero-determinant strategies in each channel.
We then show that the existence of zero-determinant strategies in multichannel games requires the existence of zero-determinant strategies in some channels.
This result implies that the existence of zero-determinant strategies in multichannel games is tightly restricted by structure of games played in each channel.
\end{abstract}


\end{frontmatter}


\section{Introduction}
\label{sec:intro}
Theory of payoff control has gradually attracted attention in repeated games in this decade.
This originates from the discovery of zero-determinant (ZD) strategies in the repeated prisoner's dilemma game \cite{PreDys2012}.
A ZD strategy is a class of strategies in repeated games which unilaterally enforce linear relationships between payoffs of players.
ZD strategies in the repeated prisoner's dilemma game contain the equalizer strategy \cite{BNS1997}, which unilaterally sets the opponent's payoff, the extortionate strategy \cite{PreDys2012}, which is unbeatable, and the generous ZD strategy \cite{StePlo2013}, which cannot win but enforces the cooperative Nash equilibrium.
Because it had been commonly accepted that such unilateral control of payoffs is impossible, the discovery of ZD strategies was received with surprise.
After the discovery, ZD strategies have been extended to broader situations \cite{HRZ2015,McAHau2016,HDND2016,McAHau2017,MamIch2020,Ued2022c}, and general mathematical properties of ZD strategies have been specified \cite{UedTan2020,TSWW2021,Ued2022b,CheLi2024}.
Particularly, Akin investigated more general strategies which enforce payoffs into some regions (not limited to hyperplanes) in the repeated prisoner's dilemma game \cite{Aki2016}.
Such strategies are more general than ZD strategies, and can be called payoff-controlling strategies.
Hilbe et al. further extended the Akin's result and defined several important classes of strategies in the context of direct reciprocity, which also allow regional payoff control \cite{HTS2015}.
Hao et al. introduced a method of regional payoff control in which the regions are realized by linear boundaries \cite{HLZ2018}.
In addition, Li and Hao extended regional payoff control to the repeated public goods game \cite{LiHao2019}.
See Ref. \cite{HCN2018} for a review of ZD strategies and more general payoff-controlling strategies.
Theory of payoff control is useful because it does not assume that all opponents are enough rational, in contrast to the classical game theory.

Recently, a new class of repeated games, called multichannel games, was introduced \cite{DHNH2020}.
A multichannel game consists of multiple repeated games simultaneously played by the same players, and players' choices in one channel can affect decisions in other channels.
Multichannel games model more realistic situations, where people are involved in multiple games in parallel, than traditional (one-channel) repeated games.
For example, appliance manufacturers compete in refrigerator market as well as in washing machine market.
It should be also mentioned that multichannel games have common characteristics with multigames \cite{WSP2014}, where different games are played in parallel in structured populations.
Donahue et al. showed that linkage between channels promotes cooperation in the multichannel donation game, because players can use all channels in order to punish the opponent's defection in one channel \cite{DHNH2020}.
Evolution of cooperation in multichannel social dilemma games has been extensively investigated \cite{WCSet2024,QCW2024,BasSen2024}.
A multichannel game is also placed in one of concurrent games \cite{RHH2025}, where opponents in each channel can be different.

Quite recently, theory of payoff control was extended to multichannel games \cite{SCFet2025}.
In Ref. \cite{SCFet2025}, Shi et al. specified strategies controlling the upper bound of the opponent's payoff, by extending the results of Hao et al. \cite{HLZ2018}.
However, their results fundamentally seem to use the existence of payoff-controlling strategies in some channels.
Existence of completely nontrivial payoff-controlling strategies in multichannel games, which do not require the existence of payoff-controlling strategies in some channels, is still unclear.
Because theory of ZD strategies underlies theory of payoff control \cite{HCN2018}, establishing theory of ZD strategies in multichannel games is necessary for further application of payoff control to the problem of evolution of cooperation in multichannel games.

The purpose of this paper is clarifying properties of ZD strategies in multichannel games.
First, we extend ZD strategies to multichannel games, and relate the existence condition of ZD strategies in multichannel games to that of ZD strategies in each channel.
We then show that equalizer strategies specific to multichannel games, which unilaterally set the opponent's payoff in two-player games, do not exist except for trivial cases, which require the existence of equalizer strategies in some channels.
This result can be extended to general ZD strategies.
Furthermore, we also show that, if a fair ZD strategy, which unilaterally equalizes payoffs of two players, exists in a multichannel two-player symmetric game, then a fair ZD strategy must exist in every channel.
These results imply that the existence of ZD strategies in multichannel games is tightly restricted by structure of games played in each channel.
We also provide several examples of ZD strategies in multichannel games.

This paper is organized as follows.
In Section \ref{sec:preliminaries}, we introduce multichannel games and ZD strategies.
In Section \ref{sec:results}, we provide our main results related to the existence of ZD strategies in multichannel games.
Section \ref{sec:conclusion} is devoted to discussion and conclusion.

\section{Preliminaries}
\label{sec:preliminaries}

\subsection{Multichannel game}
A \emph{multichannel game} consists of $M\geq 1$ repeated games \cite{DHNH2020}.
Formally, a stage game of a multichannel game is defined as $G:=\left\{ G^{(\mu)} \right\}_{\mu=1}^M$ with $G^{(\mu)} :=\left( \mathcal{N}, \left\{ A_j^{(\mu)}  \right\}_{j\in \mathcal{N}}, \left\{ s_j^{(\mu)} \right\}_{j\in \mathcal{N}} \right)$ $(\mu=1, \cdots, M)$, where $M\geq 1$ is the number of channels, $\mathcal{N}:=\{ 1, \cdots, N \}$ is the set of $N\geq 1$ players, $A_j^{(\mu)}$ is the set of actions of player $j\in \mathcal{N}$ in channel $\mu$, and $s_j^{(\mu)}: \prod_{j\in \mathcal{N}}A_j^{(\mu)} \rightarrow \mathbb{R}$ is the payoff function of player $j\in \mathcal{N}$ in channel $\mu$.
We introduce the set of action profiles in channel $\mu$ by $\mathcal{A}^{(\mu)} :=\prod_{j\in \mathcal{N}} A_j^{(\mu)}$ and the set of action profiles by $\mathcal{A} :=\prod_{\mu=1}^M \mathcal{A}^{(\mu)}$.
Similarly, we introduce an action profile in channel $\mu$ by $\bm{a}^{(\mu)} \in \mathcal{A}^{(\mu)}$ and an action profile by $\bm{a} =\left( \bm{a}^{(1)}, \cdots, \bm{a}^{(M)} \right) \in \mathcal{A}$.
For simplicity, we collectively write an action and the set of actions of player $j \in \mathcal{N}$ as $\vec{a}_j:= \left( a_j^{(1)}, \cdots, a_j^{(M)} \right)$ and $\mathcal{A}_j := \prod_{\mu=1}^M A_j^{(\mu)}$, respectively.
Therefore, an action profile $\bm{a}$ can be written in two forms, $\bm{a}=\left( \bm{a}^{(1)}, \cdots, \bm{a}^{(M)} \right)=\left( \vec{a}_1, \cdots, \vec{a}_N \right)$.
Furthermore, we also introduce the standard notation $\vec{a}_{-j}:= \left( \vec{a}_1, \cdots, \vec{a}_{j-1}, \vec{a}_{j+1}, \cdots \vec{a}_N \right)$ for the action profile of players other than $j$.
When an action profile is $\bm{a}$, the total payoff of player $j$ in a one-shot game is given by
\begin{align}
 \tilde{s}_j \left( \bm{a} \right) &:= \sum_{\mu=1}^M s_j^{(\mu)} \left( \bm{a}^{(\mu)} \right).
 \label{eq:payoff_stage}
\end{align}
In this paper, we assume that $A_j^{(\mu)}$ is a finite set for all $j\in \mathcal{N}$ and all $\mu \in \{ 1, \cdots, M \}$.

A multichannel game is defined by the infinitely repeated version of $G$.
In repeated games \cite{MaiSam2006}, a player can choose an action at each round by using all information of action profiles realized in past.
We call such plans of players as \emph{behavior strategies}.
When we write an action of player $j$ in channel $\mu$ at round $t\geq 1$ as $a_j^{(\mu)}(t)$, the payoff of player $j$ in the multichannel game is defined by
\begin{align}
 \tilde{\mathcal{S}}_j &:= \lim_{T\rightarrow \infty} \frac{1}{T} \sum_{t=1}^T \mathbb{E} \left[ \tilde{s}_j\left( \bm{a}(t) \right) \right],
 \label{eq:payoff_repeated}
\end{align}
where $\mathbb{E}[B]$ represents the expectation of the quantity $B$ with respect to behavior strategies of all players.
For convenience, we write a probability distribution of an action profile at the $t$-th round as $P_t \left( \bm{a} \right)$.
Furthermore, we introduce the limit probability distribution
\begin{align}
 P^* \left( \bm{a} \right) &:= \lim_{T\rightarrow \infty} \frac{1}{T} \sum_{t=1}^T P_{t} \left( \bm{a} \right)
 \label{eq:P_star}
\end{align}
and define the expectation of the quantity $D:\mathcal{A}\rightarrow \mathbb{R}$ with respect to $P^*$ by $\left\langle D \right\rangle^{*}$.
We remark that the symbol $^*$ represents the quantity related to the limit probability distribution $P^*$.
Then, the payoff in a multichannel game is rewritten as $\tilde{\mathcal{S}}_j=\left\langle \tilde{s}_j \right\rangle^{*}$ for all $j \in \mathcal{N}$.

\subsection{Zero-determinant strategies}
One of the simplest classes of behavior strategies is \emph{memory-one} strategies \cite{Leh1988}.
A memory-one strategy of player $j$ is given by time-independent conditional probabilities $\left\{ T_j\left( \vec{a}_j | \bm{a}^\prime \right) \right\}_{\vec{a}_j\in \mathcal{A}_j, \bm{a}^\prime \in \mathcal{A}}$, where $T_j\left( \vec{a}_j | \bm{a}^\prime \right)$ is probability of taking action $\vec{a}_j$ when an action profile in the previous round was $\bm{a}^\prime$.
When each channel is unlinked to the other channels, a memory-one strategy becomes in the form
\begin{align}
 T_j\left( \vec{a}_j | \bm{a}^\prime \right) &= \prod_{\mu=1}^M T_j^{(\mu)} \left( a_j^{(\mu)} | \bm{a}^{\prime(\mu)} \right).
\end{align}
In multichannel games, players generally can choose actions by using information from other channels.

For memory-one strategies, the following result is known as the Akin's lemma \cite{Aki2016} in multichannel games.
\begin{lemma}[\cite{SCFet2025}]
\label{lem:Akin}
For memory-one strategies of player $j$, when we define
\begin{align}
 \hat{T}_j\left( \vec{a}_j | \bm{a}^\prime \right) &:= T_j\left( \vec{a}_j | \bm{a}^\prime \right) - \delta_{\vec{a}_j, \vec{a}_j^\prime},
\end{align}
then the equality
\begin{align}
 \sum_{\bm{a}^\prime} \hat{T}_j\left( \vec{a}_j | \bm{a}^\prime \right) P^* \left( \bm{a}^\prime \right) &= 0
\end{align}
holds for all $\vec{a}_j \in \mathcal{A}_j$.
\end{lemma}
Therefore, if $\hat{T}_j$ is related to payoffs, player $j$ can unilaterally control payoffs.
This is an underlying idea of zero-determinant strategies \cite{PreDys2012,McAHau2016,UedTan2020}.
\begin{definition}
\label{def:ZDS}
A memory-one strategy of player $j$ is a \emph{zero-determinant (ZD) strategy} when it can be written in the form
\begin{align}
 \sum_{\vec{a}_j} c_{\vec{a}_j} \hat{T}_j\left( \vec{a}_j | \bm{a}^{\prime} \right) &= \sum_{k\in \mathcal{N}} \alpha_k \tilde{s}_k \left( \bm{a}^{\prime} \right) + \alpha_0 \quad \left( \forall \bm{a}^{\prime} \in \mathcal{A} \right)
 \label{eq:ZDS}
\end{align}
with some coefficients $\left\{ c_{\vec{a}_j} \right\}$ and $\left\{ \alpha_k \right\}$, and the both sides are not identically zero.
\end{definition}
As a direct consequence of Lemma \ref{lem:Akin}, when we apply $\sum_{\bm{a}^\prime} P^* \left( \bm{a}^\prime \right)$ to the both sides of Eq. (\ref{eq:ZDS}), we obtain
\begin{align}
 0 &= \sum_{k\in \mathcal{N}} \alpha_k \left\langle \tilde{s}_k \right\rangle^{*} + \alpha_0.
 \label{eq:linear}
\end{align}
That is, a ZD strategy (\ref{eq:ZDS}) unilaterally enforces a linear relation between payoffs.
Below we define $\tilde{B}\left( \bm{a} \right) := \sum_{k\in \mathcal{N}} \alpha_k \tilde{s}_k \left( \bm{a} \right) + \alpha_0$.
When a ZD strategy of player $j$ unilaterally enforces $\left\langle \tilde{B} \right\rangle^{*}=0$, we call such ZD strategy a ZD strategy controlling $\tilde{B}$.

Recently, a necessary and sufficient condition for the existence of ZD strategies was specified when action sets of all players are finite sets \cite{Ued2022b}.
Because we here consider multichannel games where action sets of all players are finite sets, we can apply this result to our problem.
\begin{proposition}[\cite{Ued2022b}]
\label{prop:existence}
A ZD strategy of player $j$ controlling $\tilde{B}$ exists if and only if there exist two actions $\overline{\vec{a}}_j, \underline{\vec{a}}_j \in \mathcal{A}_j$ of player $j$ such that
\begin{align}
 \tilde{B} \left( \overline{\vec{a}}_j, \vec{a}_{-j} \right) &\geq 0 \quad \left( \forall \vec{a}_{-j} \right) \nonumber \\
 \tilde{B} \left( \underline{\vec{a}}_j, \vec{a}_{-j} \right) &\leq 0 \quad \left( \forall \vec{a}_{-j} \right),
 \label{eq:condition_exsitence}
\end{align}
and $\tilde{B}$ is not identically zero.
\end{proposition}
This condition is intuitive: A ZD strategy stochastically switches two actions in (\ref{eq:condition_exsitence}) in order to unilaterally adjust total $\tilde{B}$.
We call the condition (\ref{eq:condition_exsitence}) in Proposition \ref{prop:existence} as the \emph{autocratic condition}.
Obviously, if a memory-one strategy of player $j$ takes more $\overline{\vec{a}}_j$ than a ZD strategy controlling $\tilde{B}$, the strategy can unilaterally control payoffs so as to  $\left\langle \tilde{B} \right\rangle^{*}\geq 0$.
Regional payoff control also implicitly uses the autocratic condition, since it uses specific types of strategies inspired from ZD strategies in social dilemma games \cite{HTS2015,Aki2016,HLZ2018}.

Finally, we remark that the autocratic condition (\ref{eq:condition_exsitence}) can be rewritten in the minimax form \cite{UedFuj2025}
\begin{align}
 \max_{\vec{a}_{j}} \min_{\vec{a}_{-j}} \tilde{B} \left( \vec{a}_j, \vec{a}_{-j} \right) &\geq 0 \nonumber \\
 \min_{\vec{a}_{j}} \max_{\vec{a}_{-j}} \tilde{B} \left( \vec{a}_j, \vec{a}_{-j} \right) &\leq 0.
 \label{eq:condition_exsitence_mod}
\end{align}
In Appendix \ref{app:minimax}, we provide the derivation of Eq. (\ref{eq:condition_exsitence_mod}).
This expression is simpler to use in the problems below.

\section{Results}
\label{sec:results}

\subsection{Existence condition}
In this subsection, we rewrite the autocratic condition (\ref{eq:condition_exsitence_mod}) into a simper form.
Below we introduce the notation
\begin{align}
 \tilde{B} \left( \bm{a} \right) &= \sum_{\mu=1}^M B^{(\mu)} \left( \bm{a}^{(\mu)} \right) + \alpha_0 \nonumber \\
 B^{(\mu)} \left( \bm{a}^{(\mu)} \right) &:= \sum_{k\in \mathcal{N}} \alpha_k s_k^{(\mu)} \left( \bm{a}^{(\mu)} \right),
 \label{eq:B}
\end{align}
and assume that each $B^{(\mu)}$ is not identically a constant.
\begin{proposition}
\label{prop:existence_MC}
The autocratic condition (\ref{eq:condition_exsitence_mod}) in multichannel games can be rewritten as
\begin{align}
 \sum_{\mu=1}^M \max_{a_{j}^{(\mu)}} \min_{a_{-j}^{(\mu)}} B^{(\mu)} \left( a_j^{(\mu)}, a_{-j}^{(\mu)} \right) + \alpha_0 &\geq 0 \nonumber \\
 \sum_{\mu=1}^M \min_{a_{j}^{(\mu)}} \max_{a_{-j}^{(\mu)}} B^{(\mu)} \left( a_j^{(\mu)}, a_{-j}^{(\mu)} \right) + \alpha_0 &\leq 0.
 \label{eq:condition_exsitence_MC}
\end{align}
\end{proposition}

\begin{proof}
According to Eq. (\ref{eq:B}), $\tilde{B}$ consists of the sum of mutually independent terms, the $\mu$-th term of which is a function of only $\bm{a}^{(\mu)}$.
Therefore, the maximization and the minimization in the autocratic condition can be calculated independently in each $\mu$ as
\begin{align}
 \max_{\vec{a}_{j}} \min_{\vec{a}_{-j}} \tilde{B} \left( \vec{a}_j, \vec{a}_{-j} \right) &= \max_{\vec{a}_{j}} \min_{\vec{a}_{-j}} \left[ \sum_{\mu=1}^M B^{(\mu)} \left( \bm{a}^{(\mu)} \right) + \alpha_0 \right]  \nonumber \\
 &= \max_{\left( a_j^{(1)}, \cdots, a_j^{(M)} \right)} \min_{\left( a_{-j}^{(1)}, \cdots, a_{-j}^{(M)} \right)} \left[ B^{(1)} \left( a_j^{(1)}, a_{-j}^{(1)} \right) + \cdots + B^{(M)} \left( a_j^{(M)}, a_{-j}^{(M)} \right) \right] + \alpha_0  \nonumber \\
 &= \sum_{\mu=1}^M \max_{a_{j}^{(\mu)}} \min_{a_{-j}^{(\mu)}} B^{(\mu)} \left( a_j^{(\mu)}, a_{-j}^{(\mu)} \right) + \alpha_0.
\end{align}
The second inequality in Eq. (\ref{eq:condition_exsitence_MC}) is obtained by the similar calculation.
\end{proof}

Proposition \ref{prop:existence_MC} results from the fact that $\tilde{B}$ can be decomposed into the sum of independent terms for each $\mu$ in multichannel games.
Although this proposition is almost trivial, it is useful to analyze the existence of ZD strategies in multichannel games below.

It is noteworthy that the existence condition of a ZD strategy controlling $B^{(\mu)}+\alpha_0^{(\mu)}$ in channel $\mu$ is given by \cite{Ued2022b}
\begin{align}
 \max_{a_{j}^{(\mu)}} \min_{a_{-j}^{(\mu)}} B^{(\mu)} \left( a_j^{(\mu)}, a_{-j}^{(\mu)} \right) + \alpha_0^{(\mu)} &\geq 0 \nonumber \\
 \min_{a_{j}^{(\mu)}} \max_{a_{-j}^{(\mu)}} B^{(\mu)} \left( a_j^{(\mu)}, a_{-j}^{(\mu)} \right) + \alpha_0^{(\mu)} &\leq 0.
 \label{eq:condition_exsitence_single}
\end{align}
Therefore, if a ZD strategy controlling $B^{(\mu)}+\alpha_0^{(\mu)}$ exists for every channel $\mu$, a ZD strategy controlling $\tilde{B}$ with $\alpha_0= \sum_{\mu=1}^M \alpha_0^{(\mu)}$ in a multichannel game also exists, since Eq. (\ref{eq:condition_exsitence_MC}) is satisfied.
However, the autocratic condition (\ref{eq:condition_exsitence_MC}) can be satisfied even if ZD strategies do not exist in some channels.

\subsection{Equalizer strategies in multichannel prisoner's dilemma game}
\label{subsec:equalizer_PD}
As an example, we consider a multichannel prisoner's dilemma game \cite{DHNH2020}, where $\mathcal{N}=\{ 1, 2 \}$, $A_j^{(\mu)}=\{ C, D \}$ $(j=1, 2; \mu\in \{1, \dots, M \})$, and $s_j^{(\mu)}$ is given by
\begin{align}
 \left( s_1^{(\mu)}(C,C), s_1^{(\mu)}(C,D), s_1^{(\mu)}(D,C), s_1^{(\mu)}(D,D) \right) &= \left( R^{(\mu)}, S^{(\mu)}, T^{(\mu)}, P^{(\mu)} \right) \nonumber \\
 \left( s_2^{(\mu)}(C,C), s_2^{(\mu)}(C,D), s_2^{(\mu)}(D,C), s_2^{(\mu)}(D,D) \right) &= \left( R^{(\mu)}, T^{(\mu)}, S^{(\mu)}, P^{(\mu)} \right)
 \label{eq:payoff_PD}
\end{align}
for all $\mu \in \{1, \dots, M \}$.
The payoffs satisfy $T^{(\mu)}>R^{(\mu)}>P^{(\mu)}>S^{(\mu)}$ and $2R^{(\mu)}>T^{(\mu)}+S^{(\mu)}$ for each $\mu \in \{1, \dots, M \}$.
The actions $C$ and $D$ represent cooperation and defection, respectively.

Here, we consider the equalizer strategy \cite{BNS1997,PreDys2012}, which unilaterally sets the payoff of the opponent.
In order to consider such strategy of player $1$, we need to set $\tilde{B}$ as
\begin{align}
 \tilde{B} \left( \bm{a} \right) &= \tilde{s}_2\left( \bm{a} \right) - r \nonumber \\
 &=  \sum_{\mu=1}^M s_2^{(\mu)} \left( \bm{a}^{(\mu)} \right) - r
\end{align}
with $r\in \mathbb{R}$.
The left-hand sides of Eq. (\ref{eq:condition_exsitence_MC}) are 
\begin{align}
 \sum_{\mu=1}^M \max_{a_{1}^{(\mu)}} \min_{a_{2}^{(\mu)}} s_2^{(\mu)} \left( a_1^{(\mu)}, a_{2}^{(\mu)} \right) - r &= \sum_{\mu=1}^M R^{(\mu)} - r
\end{align}
and 
\begin{align}
 \sum_{\mu=1}^M \min_{a_{1}^{(\mu)}} \max_{a_{2}^{(\mu)}} s_2^{(\mu)} \left( a_1^{(\mu)}, a_{2}^{(\mu)} \right) - r &= \sum_{\mu=1}^M P^{(\mu)} - r.
\end{align}
Therefore, according to Proposition \ref{prop:existence_MC}, equalizer strategies exist only for
\begin{align}
 \sum_{\mu=1}^M P^{(\mu)} &\leq r \leq \sum_{\mu=1}^M R^{(\mu)}.
 \label{eq:equalizer_PD}
\end{align}

This result is essentially the same as one reported in Ref. \cite{SCFet2025}.
However, this result is trivial, because when the inequality (\ref{eq:equalizer_PD}) is satisfied, we can adopt an equalizer strategy in each channel independently \cite{BNS1997,Ued2022b}.
Therefore, this result does not use any special properties of multichannel games.

\subsection{Combining two different games}
\label{subsec:PD_MP}
As another example, we consider a two-channel game, where channel $1$ is the prisoner's dilemma game (\ref{eq:payoff_PD}) and channel $2$ is the matching pennies game.
The matching pennies game \cite{FudTir1991} is one of the simplest zero-sum games given by $\mathcal{N}=\{ 1, 2 \}$, $A_j^{(2)}=\{ 1, 2 \}$ $(j=1, 2)$, and
\begin{align}
 \left( s_1^{(2)}(1,1), s_1^{(2)}(1,2), s_1^{(2)}(2,1), s_1^{(2)}(2,2) \right) &= \left( 1, -1, -1, 1 \right) \nonumber \\
 \left( s_2^{(2)}(1,1), s_2^{(2)}(1,2), s_2^{(2)}(2,1), s_2^{(2)}(2,2) \right) &= \left( -1, 1, 1, -1 \right).
 \label{eq:MP_2}
\end{align}

Here we again consider the equalizer strategies.
In order to consider an equalizer strategy of player $1$, we need to set $\tilde{B}$ as
\begin{align}
 \tilde{B} \left( \bm{a} \right) &= \tilde{s}_2\left( \bm{a} \right) - r \nonumber \\
 &=  \sum_{\mu=1}^2 s_2^{(\mu)} \left( \bm{a}^{(\mu)} \right) - r
\end{align}
with $r\in \mathbb{R}$.
The left-hand sides of Eq. (\ref{eq:condition_exsitence_MC}) are 
\begin{align}
 \sum_{\mu=1}^2 \max_{a_{1}^{(\mu)}} \min_{a_{2}^{(\mu)}} s_2^{(\mu)} \left( a_1^{(\mu)}, a_{2}^{(\mu)} \right) - r &= R^{(1)} - 1 - r
\end{align}
and 
\begin{align}
 \sum_{\mu=1}^2 \min_{a_{1}^{(\mu)}} \max_{a_{2}^{(\mu)}} s_2^{(\mu)} \left( a_1^{(\mu)}, a_{2}^{(\mu)} \right) - r &= P^{(1)} + 1 - r.
\end{align}
Therefore, according to Proposition \ref{prop:existence_MC}, equalizer strategies exist only for
\begin{align}
 P^{(1)} + 1 &\leq r \leq R^{(1)} - 1.
\end{align}
For such $r$ to exist, the payoffs must satisfy $P^{(1)} \leq R^{(1)} - 2$.
For example, when $\left( R^{(1)}, S^{(1)}, T^{(1)}, P^{(1)} \right)=(3, 0, 5, 1)$, an equalizer strategy with $r=2$ exists.

In order to check the validity of this result, we perform a numerical simulation for this equalizer strategy.
In Figure \ref{fig:linear_combine}, we display performance of the equalizer strategy of player $1$ against $1000$ randomly generated memory-one strategies of player $2$ for this parameter value.
\begin{figure}
\begin{center}
\includegraphics[clip, width=8.0cm]{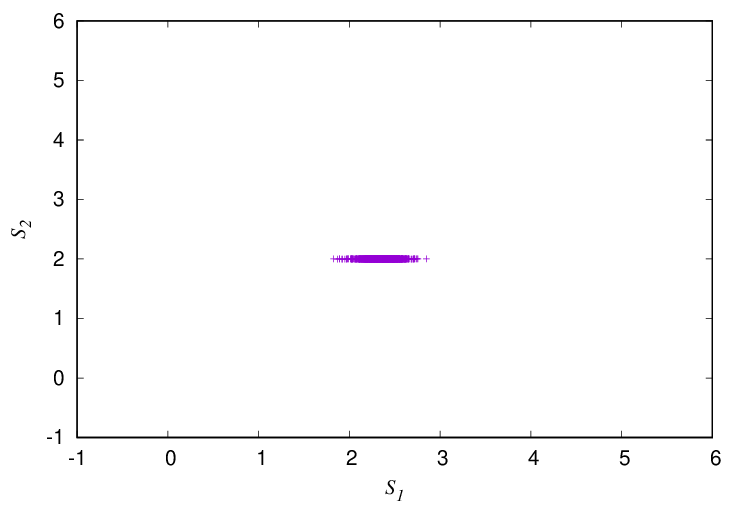}
\end{center}
\caption{A linear relation between $\left\langle \tilde{s}_1 \right\rangle^{*}$ and $\left\langle \tilde{s}_2 \right\rangle^{*}$ when player $1$ uses the equalizer strategy with $r=2$ and player $2$ uses $1000$ randomly generated memory-one strategies. The payoffs in channel $1$ are set to $\left( R^{(1)}, S^{(1)}, T^{(1)}, P^{(1)} \right)=(3, 0, 5, 1)$, and the payoffs in channel $2$ are set as in Eq. (\ref{eq:MP_2}). Each $\left\langle \tilde{s}_j \right\rangle^{*}$ is calculated by time average over $10^6$ time steps. The equalizer strategy indeed enforces a linear relation $\left\langle \tilde{s}_2 \right\rangle^{*}=2$.}
\label{fig:linear_combine}
\end{figure}
Details of construction of the ZD strategy are provided in Appendix \ref{app:simulation}.
We find that the equalizer strategy effectively enforces a linear relation $\left\langle \tilde{s}_2 \right\rangle^{*}=2$.

It should be noted that the matching pennies game itself does not contain any equalizer strategies \cite{UedFuj2025}.
However, by combining two different games, equalizer strategies become to exist.
Although such situation was conjectured in the previous study \cite{SCFet2025}, it was not explicitly shown.

\subsection{Absence of nontrivial equalizer strategies in multichannel games}
Next, we prove that nontrivial equalizer strategies do not exist in multichannel games.
We set $\mathcal{N}=\{ 1, 2 \}$, and consider equalizer strategies of player $1$.
We call an equalizer strategy in a multichannel game \emph{nontrivial} if no equalizer strategies exist in all channels but they exist in the multichannel game.

We first show the following lemma.
\begin{lemma}
\label{lem:existence_equalizer}
Equalizer strategies of player $1$ exist in channel $\mu$ if and only if the inequality
\begin{align}
 \max_{a_{1}^{(\mu)}} \min_{a_{2}^{(\mu)}} s_2^{(\mu)} \left( a_1^{(\mu)}, a_2^{(\mu)} \right) &\geq \min_{a_{1}^{(\mu)}} \max_{a_{2}^{(\mu)}} s_2^{(\mu)} \left( a_1^{(\mu)}, a_2^{(\mu)} \right)
 \label{eq:existence_equalizer_single}
\end{align}
holds.
\end{lemma}

\begin{proof}
If equalizer strategies of player $1$ exist, according to Eq. (\ref{eq:condition_exsitence_single}), there exists a real number $r^{(\mu)}$ such that
\begin{align}
 \max_{a_{1}^{(\mu)}} \min_{a_{2}^{(\mu)}} s_2^{(\mu)} \left( a_1^{(\mu)}, a_2^{(\mu)} \right) - r^{(\mu)} &\geq 0 \nonumber \\
 \min_{a_{1}^{(\mu)}} \max_{a_{2}^{(\mu)}} s_2^{(\mu)} \left( a_1^{(\mu)}, a_2^{(\mu)} \right) - r^{(\mu)} &\leq 0.
\end{align}
This means
\begin{align}
 \max_{a_{1}^{(\mu)}} \min_{a_{2}^{(\mu)}} s_2^{(\mu)} \left( a_1^{(\mu)}, a_2^{(\mu)} \right) &\geq r^{(\mu)} \geq \min_{a_{1}^{(\mu)}} \max_{a_{2}^{(\mu)}} s_2^{(\mu)} \left( a_1^{(\mu)}, a_2^{(\mu)} \right)
\end{align}
and the inequality (\ref{eq:existence_equalizer_single}) holds.

Conversely, if the inequality (\ref{eq:existence_equalizer_single}) holds, we can choose a real number $r^{(\mu)}$ such that inequalities
\begin{align}
 \max_{a_{1}^{(\mu)}} \min_{a_{2}^{(\mu)}} s_2^{(\mu)} \left( a_1^{(\mu)}, a_2^{(\mu)} \right) &\geq r^{(\mu)} \geq \min_{a_{1}^{(\mu)}} \max_{a_{2}^{(\mu)}} s_2^{(\mu)} \left( a_1^{(\mu)}, a_2^{(\mu)} \right)
\end{align}
hold.
This implies the existence condition (\ref{eq:condition_exsitence_single}) of equalizer strategies
\begin{align}
 \max_{a_{1}^{(\mu)}} \min_{a_{2}^{(\mu)}} s_2^{(\mu)} \left( a_1^{(\mu)}, a_2^{(\mu)} \right) - r^{(\mu)} &\geq 0 \nonumber \\
 \min_{a_{1}^{(\mu)}} \max_{a_{2}^{(\mu)}} s_2^{(\mu)} \left( a_1^{(\mu)}, a_2^{(\mu)} \right) - r^{(\mu)} &\leq 0
\end{align}
holds.
\end{proof}

By using this lemma, we prove that nontrivial equalizer strategies do not exist in multichannel games.
\begin{theorem}
\label{thm:absence_equalizer}
If equalizer strategies do not exist in all channels, then equalizer strategies do not exist in the multichannel game.
\end{theorem}

\begin{proof}
If equalizer strategies of player $1$ do not exist in all channels, according to Lemma \ref{lem:existence_equalizer}, we obtain
\begin{align}
 \max_{a_{1}^{(\mu)}} \min_{a_{2}^{(\mu)}} s_2^{(\mu)} \left( a_1^{(\mu)}, a_2^{(\mu)} \right) &< \min_{a_{1}^{(\mu)}} \max_{a_{2}^{(\mu)}} s_2^{(\mu)} \left( a_1^{(\mu)}, a_2^{(\mu)} \right)
\end{align}
for all $\mu$.
By considering the sum of both sides with respect to $\mu$, we obtain
\begin{align}
 \sum_{\mu=1}^M \max_{a_{1}^{(\mu)}} \min_{a_{2}^{(\mu)}} s_2^{(\mu)} \left( a_1^{(\mu)}, a_2^{(\mu)} \right) &< \sum_{\mu=1}^M \min_{a_{1}^{(\mu)}} \max_{a_{2}^{(\mu)}} s_2^{(\mu)} \left( a_1^{(\mu)}, a_2^{(\mu)} \right).
\end{align}
Then, we find that $\alpha_0$ in Proposition \ref{prop:existence_MC} does not exist, and conclude that equalizer strategies of player $1$ do not exist in the multichannel game.
\end{proof}

The contraposition of Theorem \ref{thm:absence_equalizer} is ``If an equalizer strategy exists in a multichannel game, then an equalizer strategy exists at least in one channel.''.
We remark that the example in subsection \ref{subsec:PD_MP} corresponds to a trivial case, because the prisoner's dilemma game contains equalizer strategies.
According to Lemma \ref{lem:existence_equalizer}, equalizer strategies of player $1$ in channel $\mu$ exist if and only if the margin
\begin{align}
 \Delta^{(\mu)} &:= \max_{a_{1}^{(\mu)}} \min_{a_{2}^{(\mu)}} s_2^{(\mu)} \left( a_1^{(\mu)}, a_2^{(\mu)} \right) - \min_{a_{1}^{(\mu)}} \max_{a_{2}^{(\mu)}} s_2^{(\mu)} \left( a_1^{(\mu)}, a_2^{(\mu)} \right)
\end{align}
satisfies $\Delta^{(\mu)} \geq 0$.
Similarly, according to the proof of Theorem \ref{thm:absence_equalizer}, equalizer strategies in the multichannel game exist if and only if the total margin $\sum_{\mu=1}^M \Delta^{(\mu)}$ satisfies
\begin{align}
 \sum_{\mu=1}^M \Delta^{(\mu)} &\geq 0.
\end{align}
This quantity enables us to easily judge whether equalizer strategies exist in the multichannel game or not.
For example, for the example in Section \ref{subsec:equalizer_PD}, we obtain
\begin{align}
 \sum_{\mu=1}^M \Delta^{(\mu)} &= \sum_{\mu=1}^M R^{(\mu)} - \sum_{\mu=1}^M P^{(\mu)}  > 0.
\end{align}
In contrast, for the example in Section \ref{subsec:PD_MP}, we obtain
\begin{align}
 \sum_{\mu=1}^2 \Delta^{(\mu)} &= R^{(1)} - 1 - \left( P^{(1)} + 1 \right),
\end{align}
which is nonnegative if and only if $P^{(1)} \leq R^{(1)} - 2$.

It is also noteworthy that the results in this subsection can be easily extended to general ZD strategies, by replacing $s_2^{(\mu)}$ and $\left( a_1^{(\mu)}, a_2^{(\mu)} \right)$ with $B^{(\mu)}$ and $\left( a_j^{(\mu)}, a_{-j}^{(\mu)} \right)$, respectively.
\begin{theorem}
\label{thm:absence_ZD}
If ZD strategies of player $j$ setting the values of $B^{(\mu)}$ do not exist in all channels $\mu$, then ZD strategies of player $j$ setting the values of $\tilde{B}$ do not exist in the multichannel game.
\end{theorem}

\subsection{Absence of nontrivial fair ZD strategies in multichannel two-player symmetric games}
As another application of Proposition \ref{prop:existence_MC}, we prove that nontrivial fair ZD strategies do not exist in multichannel two-player symmetric games.
A two-player game is \emph{symmetric} if the payoffs satisfy $s_2^{(\mu)} \left( a_1^{(\mu)} , a_2^{(\mu)}  \right) = s_1^{(\mu)}  \left( a_2^{(\mu)} , a_1^{(\mu)}  \right)$ for all $\left( a_1^{(\mu)} , a_2^{(\mu)}  \right)$ and all $\mu$.
A ZD strategy is \emph{fair} in a two-player symmetric game if it unilaterally enforces 
\begin{align}
 \left\langle \tilde{s}_1 \right\rangle^{*} &= \left\langle \tilde{s}_2 \right\rangle^{*}.
 \label{eq:linear_fair}
\end{align}
An example of fair ZD strategies in repeated games is the Tit-for-Tat strategy \cite{RCO1965} in the prisoner's dilemma game \cite{PreDys2012}.
It has been known that fair ZD strategies exist in potential games \cite{Ued2022} and games with no generalized rock-paper-scissors cycles \cite{Ued2022b}.
Here we prove the following theorem.
\begin{theorem}
\label{thm:absence_fair}
If there exists a fair ZD strategy controlling $\tilde{s}_1-\tilde{s}_2$ in a multichannel game, then there must exist a fair ZD strategy controlling $s_1^{(\mu)}-s_2^{(\mu)}$ in every channel $\mu$.
\end{theorem}

\begin{proof}
According to Proposition \ref{prop:existence_MC}, a necessary and sufficient condition for the existence of a fair ZD strategy of player $1$ controlling $\tilde{s}_1-\tilde{s}_2$ in a multichannel game is
\begin{align}
 \sum_{\mu=1}^M \max_{a_{1}^{(\mu)}} \min_{a_{2}^{(\mu)}} \left[ s_1^{(\mu)} \left( a_1^{(\mu)}, a_{2}^{(\mu)} \right) - s_2^{(\mu)} \left( a_1^{(\mu)}, a_{2}^{(\mu)} \right) \right] &\geq 0 \nonumber \\
 \sum_{\mu=1}^M \min_{a_{1}^{(\mu)}} \max_{a_{2}^{(\mu)}} \left[ s_1^{(\mu)} \left( a_1^{(\mu)}, a_{2}^{(\mu)} \right) - s_2^{(\mu)} \left( a_1^{(\mu)}, a_{2}^{(\mu)} \right) \right] &\leq 0.
 \label{eq:condition_exsitence_fair}
\end{align}
We define the anti-symmetric part of the payoff by
\begin{align}
 s_\mathrm{A}^{(\mu)} \left( a_1^{(\mu)}, a_{2}^{(\mu)} \right) &:= \frac{1}{2} \left[ s_1^{(\mu)} \left( a_1^{(\mu)}, a_{2}^{(\mu)} \right) - s_2^{(\mu)} \left( a_1^{(\mu)}, a_{2}^{(\mu)} \right) \right].
\end{align}
It should be noted that the diagonal components satisfy $s_\mathrm{A}^{(\mu)}\left( a_1^{(\mu)}, a_1^{(\mu)} \right)=0$ for all $a_1^{(\mu)}$.
Then we find
\begin{align}
 \min_{a_{2}^{(\mu)}} s_\mathrm{A}^{(\mu)} \left( a_1^{(\mu)}, a_{2}^{(\mu)} \right) &\leq 0 \quad \left( \forall a_1^{(\mu)} \right) \nonumber \\
 \max_{a_{2}^{(\mu)}} s_\mathrm{A}^{(\mu)} \left( a_1^{(\mu)}, a_{2}^{(\mu)} \right) &\geq 0 \quad \left( \forall a_1^{(\mu)} \right)
\end{align}
for all $\mu$, which leads to
\begin{align}
 \max_{a_{1}^{(\mu)}} \min_{a_{2}^{(\mu)}} s_\mathrm{A}^{(\mu)} \left( a_1^{(\mu)}, a_{2}^{(\mu)} \right) &\leq 0 \nonumber \\
 \min_{a_{1}^{(\mu)}} \max_{a_{2}^{(\mu)}} s_\mathrm{A}^{(\mu)} \left( a_1^{(\mu)}, a_{2}^{(\mu)} \right) &\geq 0
\end{align}
for all $\mu$.
Therefore, if the condition (\ref{eq:condition_exsitence_fair}) is satisfied, we obtain
\begin{align}
 \max_{a_{1}^{(\mu)}} \min_{a_{2}^{(\mu)}} s_\mathrm{A}^{(\mu)} \left( a_1^{(\mu)}, a_{2}^{(\mu)} \right) &= 0 \nonumber \\
 \min_{a_{1}^{(\mu)}} \max_{a_{2}^{(\mu)}} s_\mathrm{A}^{(\mu)} \left( a_1^{(\mu)}, a_{2}^{(\mu)} \right) &= 0
\end{align}
for all $\mu$.
As we saw in Eq. (\ref{eq:condition_exsitence_single}), this implies that the existence condition of a fair ZD strategy in every channel $\mu$ is satisfied.
Thus, a fair ZD strategy controlling $s_1^{(\mu)}-s_2^{(\mu)}$ exists in every channel $\mu$.
\end{proof}

We remark that Theorem \ref{thm:absence_fair} is stronger than the contraposition of Theorem \ref{thm:absence_ZD} applied to fair ZD strategies.

\subsection{More complicated payoff control}
The autocratic condition given in this paper can be used even if $B^{(\mu)} $ is not written in the form of Eq. (\ref{eq:B}).
For example, we consider a two-channel prisoner's dilemma game (\ref{eq:payoff_PD}).
When we set 
\begin{align}
 \tilde{B} \left( \bm{a} \right) &= s_1^{(1)} \left( \bm{a}^{(1)} \right) + s_2^{(2)} \left( \bm{a}^{(2)} \right) - r
\end{align}
with $r\in \mathbb{R}$, the autocratic condition of player $1$ is Eq. (\ref{eq:condition_exsitence_MC}) with $B^{(1)}=s_1^{(1)}$ and $B^{(2)}=s_2^{(2)}$.
The left-hand sides are calculated as
\begin{align}
 \sum_{\mu=1}^2 \max_{a_{1}^{(\mu)}} \min_{a_{2}^{(\mu)}} B^{(\mu)} \left( a_1^{(\mu)}, a_2^{(\mu)} \right) - r &= P^{(1)} + R^{(2)} - r \nonumber \\
 \sum_{\mu=1}^2 \min_{a_{1}^{(\mu)}} \max_{a_{2}^{(\mu)}} B^{(\mu)} \left( a_1^{(\mu)}, a_2^{(\mu)} \right) - r &= R^{(1)} + P^{(2)} - r.
\end{align}
Therefore, according to Proposition \ref{prop:existence_MC}, when $r$ satisfies
\begin{align}
 R^{(1)} + P^{(2)} &\leq r \leq P^{(1)} + R^{(2)},
\end{align}
a ZD strategy of player $1$ unilaterally enforcing
\begin{align}
 0 &= \left\langle s_1^{(1)} \right\rangle^{*} + \left\langle s_2^{(2)} \right\rangle^{*} - r,
\end{align}
exists.
For such $r$ to exist, the inequality $R^{(1)} - P^{(1)} \leq R^{(2)} - P^{(2)}$ must be satisfied.

\section{Discussion and Conclusion}
\label{sec:conclusion}
In this paper, we investigated the existence condition of ZD strategies in multichannel games.
We found that the existence of ZD strategies in multichannel games generally requires the existence of ZD strategies in some channels.
We also showed that the existence of fair ZD strategies in multichannel games requires the existence of fair ZD strategies in all channels.
These results imply that the existence of ZD strategies in multichannel games is tightly restricted by structure of games played in each channel.
We also provide several examples of ZD strategies in multichannel games.

Extensions of our results to more complicated situations, such as games with imperfect monitoring \cite{HRZ2015}, games with discounting \cite{HTS2015}, games with alternating decision-making \cite{McAHau2017}, and games with longer memory \cite{Ued2022c}, are a subject of future work.
Furthermore, we also remark that investigation of the effectiveness of our results in games in structured populations \cite{SzoPer2014,SzoPer2014b} is a significant problem.
As noted in Section \ref{sec:intro}, multichannel games have things in common with multigames \cite{WSP2014}, which are different games played in parallel in structured populations.
Recently, a novel framework analyzing evolutionary dynamics of games on regular graphs was discovered \cite{WPS2024}.
We want to determine potential limits of our results in structured populations in future.
In addition, application of machine learning technique may be useful for analytically intractable problems \cite{IOUP2023,IOUP2025}.

The results of this paper are significant for investigating evolution of cooperation in multichannel games.
As noted in Section \ref{sec:intro}, theory of ZD strategies underlies theory of payoff control \cite{HCN2018}.
Particularly, a pair of equalizer strategies is a Nash equilibrium in two-player games, because each player has no incentive to deviate from the strategy.
In the multichannel prisoner's dilemma game, a cooperative Nash equilibrium is also realized by a pair of cooperative equalizer strategies.
Our Theorem \ref{thm:absence_equalizer} states that, for achieving such equilibrium, equalizer strategies must be used at least in one channel, and no nontrivial equalizer strategies exist.
This fact suggests that theory of payoff control in repeated games is sufficient to explain evolution of cooperation in multichannel games.
Furthermore, our Theorem \ref{thm:absence_fair} claims that, in order to construct a fair ZD strategy in multichannel two-player games, a player must use fair ZD strategies in all channels.
Fair ZD strategies, like the Tit-for-Tat strategy, are useful in evolution of cooperation, because they can invade arbitrary strategies by neutral drift, and create routes for cooperative strategies to evolve \cite{NowSig1992}.
Our results imply that players must use the Tit-for-Tat strategy in all channels in order to adopt a fair ZD strategy in the multichannel prisoner's dilemma game.
No complicated construction of fair ZD strategies is necessary for multichannel games.

\section*{Acknowledgement}
We thank Ken Nakamura for useful comments.
This study was supported by Toyota Riken Scholar Program.

\appendix

\section{Derivation of Eq. (\ref{eq:condition_exsitence_mod})}
\label{app:minimax}
We remark that Eq. (\ref{eq:condition_exsitence}) can be rewritten as
\begin{align}
 \min_{\vec{a}_{-j}} \tilde{B} \left( \overline{\vec{a}}_j, \vec{a}_{-j} \right) &\geq 0 \nonumber \\
 \max_{\vec{a}_{-j}} \tilde{B} \left( \underline{\vec{a}}_j, \vec{a}_{-j} \right) &\leq 0.
\end{align}
Therefore, we obtain
\begin{align}
 \max_{\vec{a}_{j}} \min_{\vec{a}_{-j}} \tilde{B} \left( \vec{a}_j, \vec{a}_{-j} \right) &\geq \min_{\vec{a}_{-j}} \tilde{B} \left( \overline{\vec{a}}_j, \vec{a}_{-j} \right) \geq 0 \nonumber \\
 \min_{\vec{a}_{j}} \max_{\vec{a}_{-j}} \tilde{B} \left( \vec{a}_j, \vec{a}_{-j} \right) &\leq \max_{\vec{a}_{-j}} \tilde{B} \left( \underline{\vec{a}}_j, \vec{a}_{-j} \right) \leq 0.
\end{align}

Conversely, if Eq. (\ref{eq:condition_exsitence_mod}) holds, we can define
\begin{align}
 \overline{\vec{a}}_j &:= \arg \max_{\vec{a}_{j}} \min_{\vec{a}_{-j}} \tilde{B} \left( \vec{a}_j, \vec{a}_{-j} \right) \nonumber \\
 \underline{\vec{a}}_j &:= \arg \min_{\vec{a}_{j}} \max_{\vec{a}_{-j}} \tilde{B} \left( \vec{a}_j, \vec{a}_{-j} \right)
\end{align}
as the two actions in the condition (\ref{eq:condition_exsitence}).

\section{Construction of ZD strategies in numerical simulation}
\label{app:simulation}
For the example of Section \ref{subsec:PD_MP}, according to Ref. \cite{Ued2022b}, if the autocratic condition is satisfied, a ZD strategy of player $j$ controlling $\tilde{B}$ is constructed as
\begin{align}
 T_j\left( \overline{\vec{a}}_j | \bm{a}^\prime \right) &= \delta_{\vec{a}_j^\prime, \overline{\vec{a}}_j} - \frac{1}{W} \left( \tilde{B}\left( \bm{a}^\prime \right) \mathbb{I}\left( \tilde{B}\left( \bm{a}^\prime \right)>0 \right) \delta_{\vec{a}_j^\prime, \overline{\vec{a}}_j} + \sum_{\vec{a}_j\neq \overline{\vec{a}}_j, \underline{\vec{a}}_j} \tilde{B}\left( \bm{a}^\prime \right) \mathbb{I}\left( \tilde{B}\left( \bm{a}^\prime \right)<0 \right) \delta_{\vec{a}_j^\prime, \vec{a}_j} \right) \nonumber \\
 T_j\left( \underline{\vec{a}}_j | \bm{a}^\prime \right) &= \delta_{\vec{a}_j^\prime, \underline{\vec{a}}_j} + \frac{1}{W} \left( \tilde{B}\left( \bm{a}^\prime \right) \mathbb{I}\left( \tilde{B}\left( \bm{a}^\prime \right)<0 \right) \delta_{\vec{a}_j^\prime, \underline{\vec{a}}_j} + \sum_{\vec{a}_j\neq \overline{\vec{a}}_j, \underline{\vec{a}}_j} \tilde{B}\left( \bm{a}^\prime \right) \mathbb{I}\left( \tilde{B}\left( \bm{a}^\prime \right)>0 \right) \delta_{\vec{a}_j^\prime, \vec{a}_j} \right) \nonumber \\
 T_j\left( \vec{a}_j | \bm{a}^\prime \right) &= \delta_{\vec{a}_j^\prime, \vec{a}_j} + \frac{1}{W} \left( - \left| \tilde{B}\left( \bm{a}^\prime \right) \right| \delta_{\vec{a}_j^\prime, \overline{\vec{a}}_j} - \frac{1}{L_j-2} \tilde{B}\left( \bm{a}^\prime \right) \mathbb{I}\left( \tilde{B}\left( \bm{a}^\prime \right)<0 \right) \delta_{\vec{a}_j^\prime, \underline{\vec{a}}_j} \right. \nonumber \\
 & \qquad \left. + \frac{1}{L_j-2} \tilde{B}\left( \bm{a}^\prime \right) \mathbb{I}\left( \tilde{B}\left( \bm{a}^\prime \right)>0 \right) \delta_{\vec{a}_j^\prime, \overline{\vec{a}}_j} \right)  \left( \vec{a}_j\neq \overline{\vec{a}}_j, \underline{\vec{a}}_j \right),
\end{align}
where $\mathbb{I}$ is the indicator function, $L_j:= \left| \mathcal{A}_j \right|$, and
\begin{align}
 W &:= \max_{\bm{a}\in \mathcal{A}} \left| \tilde{B}\left( \bm{a} \right) \right|.
\end{align}
For this case, $L_j=4$, and we can choose $\overline{\vec{a}}_1=(C,1)$ and $\underline{\vec{a}}_1=(D,2)$.
The quantity $W$ is calculated as $W=4$.

\section*{References}

\bibliography{ZDS}

\end{document}